\documentclass[12pt]{article}
\usepackage{enumerate}
\usepackage{amsfonts}
\usepackage{latexsym}
\usepackage{color}
\usepackage{graphicx}
\usepackage{wrapfig}
\usepackage{caption}

\usepackage{amsmath,amssymb,amsthm}

\theoremstyle{plain}
\newtheorem{theorem}{Theorem}
\newtheorem{lemma}[theorem]{Lemma}
\newtheorem{proposition}[theorem]{Proposition}
\newtheorem{corollary}[theorem]{Corollary}

\theoremstyle{definition}

\newtheorem{remark}[theorem]{Remark}

\newcommand{\HH}{\mathcal H}
\newcommand{\St}{\mathfrak S}
\newcommand{\Tr}{\operatorname{Tr}}
\newcommand{\B}{\mathfrak{B}}
\newcommand{\T}{\mathfrak{T}}
\newcommand{\PP}{\mathfrak{P}}
\newcommand{\Bsa}{\mathfrak{B}_{\mathrm{sa}}}
\newcommand{\supp}{\operatorname{supp}}

\newcommand{\rank}{\operatorname{rank}}

\newcommand{\ra}{\rightarrow}
\newcommand{\ext}{\mathrm{ext}}
\begin{document}
\title{On supporting affine functionals for Entanglement of Formation}
\author{A.S.Holevo, M.E.~Shirokov\\
Steklov Mathematical Institute, Moscow, Russia}
\date{}
\maketitle

\maketitle

\begin{abstract}

In several articles, the authors assume that the convex roof structure of the EoF and finite-dimensionality of subsystems $A$ and $B$ guarantee the existence of
the (global) supporting affine functional for the EoF at any state of the system $AB$. This means that for any state $\rho$ of $AB$ there is a Hermitian operator $\Lambda_\rho$ on $\mathcal{H}_{AB}=\mathcal{H}_A\otimes\mathcal{H}_B$
such that $E_F(\rho)=\mathrm{Tr}\Lambda_\rho\rho$ and $E_F(\sigma)\geq\mathrm{Tr}\Lambda_\rho\sigma$ for any state $\sigma$ of $AB$. We present an explicit example showing that, when $\rho$ is degenerate, this is not true even in the simplest case when $A$ and $B$ are qubit systems. The construction is based on the fact that the existence of a supporting affine functional for the EoF at a state $\rho$  is equivalent to the Lipschitz lower  semicontinuity of the EoF at this state $\rho$.  We use  Wootters' formula and the help of Claude Fable 5 to find a state $\rho$ of the system $AB$ for which the latter property does not hold.

We also describe  conditions for the existence the local and global supporting affine functionals for the EoF at a given state of both finite and infinite-dimensional bipartite quantum systems. These conditions allow us to find   Lipschitz lower semicontinuity bounds for the EoF at a given finite rank state $\rho$ (i.e. inequalities of the form $\,E_F(\rho)-E_F(\sigma)\leq C_\rho\|\rho-\sigma\|_1$) with and without restrictions on the support of the state $\sigma$.
\end{abstract}

\section{Introduction}

The Entanglement of Formation (EoF) is one of the basic entanglement measures in bipartite quantum systems \cite{P&V,4H}. If $A$ and $B$ are finite-dimensional quantum systems then the EoF of a state $\rho$ of the composite system $AB$ is defined as a convex roof extension of the function
$\rho\mapsto S(\rho_A)=S(\rho_B)$ on the set $\ext\St(\HH_{AB})$ of pure states of the system $AB$ to the set $\St(\HH_{AB})$ of all states of this system, i.e.
\begin{equation}\label{EF-d}
E_F(\rho)=\inf_{\sum_k\!p_k\rho_k=\rho}\sum_kp_kS([\rho_k]_A),
\end{equation}
where the  infimum is over all discrete (finite or countable) ensembles $\{p_k, \rho_k\}$ of pure states in $\St(\HH_{AB})$ with the average state $\rho$ \cite{Bennett}.\smallskip

It is known (cf. \cite{N}) that $E_F$ is a uniformly continuous function on the set $\St(\HH_{AB})$. So, as the function $E_F$  is convex by the definition, it is a \emph{convex closure}\footnote{The convex closure of a function $f$ on a convex subset of a metric linear space is the maximal lower semicontinuous convex function on this set majorized by $f$ \cite{Roc}.} of the concave
function $\rho\mapsto S(\rho_A)$ on $\St(\HH_{AB})$. Hence, by the Fenchel-Moreau theorem, it can be represented as a double Fenchel transform of the latter function, which can be written as
\begin{equation}\label{VE}
  E_F(\rho)=\sup\left\{\Tr\Lambda\rho\,\left|\,\Lambda\in\Bsa(\HH_{AB}),
  \langle\psi|\Lambda|\psi\rangle\le S\bigl(\psi_A\bigr)\ \forall\psi\in \HH^1_{AB}\right.\right\},
\end{equation}
where $\Bsa(\HH_{AB})$ is the set of all bounded Hermitian operators on  $\HH_{AB}$ and $\HH^1_{AB}$ is the unit sphere in $\HH_{AB}$ (we use the shortened record $\psi_A\doteq [|\psi\rangle\langle\psi|]_A$)  \cite{Ulm,H-NEW}.

It is remarkable that representation (\ref{VE}) remains valid in the case when $A$ and $B$ are infinite-dimensional quantum systems, while the
applicability of the expression (\ref{EF-d}) in this general settings is (as far as we know) an open question. There are
serious reasons (described in \cite{EM, Lami-new})  to believe that in this case  a true version of the EoF should be defined as
\begin{equation}\label{EF-c}
E_F(\rho)=\!\inf_{\int\varrho\mu(d\varrho)=\rho}\int_{\ext\St(\HH_{AB})}S(\varrho_A)\mu(d\varrho),
\end{equation}
where the  infimum is over all Borel probability measures\footnote{It is reasonable to consider such measures as generalized ensembles of quantum states \cite{H-Sh-2,H-SCI}.} on the set of pure states in $\St(\HH_{AB})$ with the barycenter $\rho$ (this infimum is always attained) \cite[Section 5]{EM}.\smallskip

It is proved in \cite{EM} that the definitions (\ref{EF-d}) and (\ref{EF-c}) coincide if at least one of the marginal entropies of the state $\rho$ is finite and that the function
$E_F$ defined in (\ref{EF-c}) is lower  semicontinuous on the set $\St(\HH_{AB})$  (while the lower semicontinuity of the function
$E_F$ defined in (\ref{EF-d}) is equivalent to its coincidence with the function $E_F$ defined in (\ref{EF-c})).
So, in this article \emph{we assume that the EoF is defined by formula  (\ref{EF-c}) keeping in mind that this definition coincide with (\ref{EF-d}) for any state  $\rho$ of a finite-dimensional bipartite system $AB$ and for any state  $\rho$ of an infinite-dimensional bipartite system $AB$ such that
$\,\min\{S(\rho_A),S(\rho_B)\}<+\infty$}.\smallskip

It is the lower  semicontinuity of the convex function $E_F$ defined in (\ref{EF-c}) that implies, by the Fenchel-Moreau theorem, the validity of the variational expression (\ref{VE}) for the EoF in
infinite-dimensional bipartite quantum systems.\smallskip

The variational expression (\ref{VE}) has a simple geometrical interpretation. Since every affine functional on $\St(\HH_{AB})$ has the form $\ell(\rho)=\Tr\Lambda\rho$ with Hermitian $\Lambda$,
expression (\ref{VE}) shows that the function $E_F$ is the poinwise supremum of the family of all affine functionals on $\St(\HH_{AB})$ not exceeding $E_F$  (any convex lower semicontinuous function  on $\St(\HH_{AB})$
has such a representation \cite{Roc}).\smallskip

This article is devoted to the question of attainability of the supremum in (\ref{VE}) for a given state $\rho$. This attainability means the existence
of a functional $\ell(\rho)=\Tr\Lambda_{\rho}\rho$  on $\St(\HH_{AB})$ such that  $\ell(\rho)=E_F(\rho)$ and $\ell(\sigma)\leq E_F(\sigma)$  for any state  $\sigma$ in $\St(\HH_{AB})$.\footnote{The condition $\,\langle\psi|\Lambda|\psi\rangle\le S\bigl(\psi_A\bigr)\,$ for all $\,\psi\in \HH^1_{AB}$ implies that $\,\Tr\Lambda_{\rho}\sigma\leq E_F(\sigma)\,$ for any state  $\sigma$ in $\St(\HH_{AB})$, see the proof of the implication $\rm (ii)\Rightarrow(i)$ in Theorem \ref{P-3}.}
Such functional is called \emph{supporting} for the function $E_F$ at the state $\rho\in\St(\HH_{AB})$.

General theorems of convex analysis imply the existence of a supporting functional for a convex continuous function on a compact convex subset of a finite-dimensional space at \emph{any  internal point} of this set: for a proper convex function on a compact convex subset of a finite-dimensional linear space, the subdifferential is nonempty at every point of the relative interior of its effective domain; see, for example, \cite[Theorem~23.4]{Roc} (subdifferential is the set of all supporting functionals).
This result  implies, in particular, that \emph{the supremum in (\ref{VE}) is always attained provided that $\rho$ is a non-degenerate state of a finite-dimensional bipartite quantum system $AB$.} A direct proof of this claim is given in \cite{H-NEW}.

At the same time, the existence of a supporting functional for a convex continuous function on a compact convex set at boundary points of this set cannot be proved in general: the convex function $\,g(x)=-\sqrt{1-x^2}\,$ on $\,[-1,1]\,$ has no finite supporting affine functional at $x=\pm1$. Thus, the question of  attainability of the supremum in (\ref{VE}) for a given degenerate state $\rho$
cannot be resolved from the general convex analysis point of view. Nevertheless, \emph{the special roof structure} of the EoF (described in \cite{Ulm}) gave us  a hope that the existence of the supporting affine functional
can be established for all states of a finite-dimensional bipartite quantum system $AB$ using additional arguments. One of the ideas in this direction was to use
the existence of \emph{local} supporting affine functional for any state (described at the end of Section 3) and construct its appropriate extension  to the set $\St(\HH_{AB})$.
The main aim of this article is to show that all these hopes are in vain.

Our approach is based on a simple criterion of attainability of the supremum in (\ref{VE}) at a given state $\rho$ described in Section 3. This criterion allows us to reduce the question of
existence (resp. nonexistence) of the supporting affine functional for the EoF at a given state $\rho$ to the question of
existence (resp. nonexistence) for this state $\rho$ of a positive number $C_{\rho}$  such that
\begin{equation}\label{LSC}
 E_F(\rho)-E_F(\sigma)\leq \frac{C_{\rho}}{2}\,\|\rho-\sigma\|_1\quad \forall\sigma\in\St(\mathcal{H}_{AB}).
\end{equation}
The validity of (\ref{LSC}) with some positive $C_{\rho}$ can be called Lipschitz semicontinuity of the EoF at the state $\rho$. It essential that the minimal $C_{\rho}$ in (\ref{LSC}) is equal to
the minimal spectral diameter of the operator $\Lambda_{\rho}$ at which the supremum in (\ref{VE}) is attained (Theorem 1).

The idea of proving the absence of the Lipschitz semicontinuity of the EoF at some degenerate state is simple: to use Wootters' formula for searching  a state $\rho$ of
the 2-qubit system $AB$ for which (\ref{LSC}) does not hold with any finite $C_{\rho}$. Nevertheless, it was not clear how to
do it practically (as  Wootters' formula is not too simple). Fortunately, this is done very quickly with Claude Fable 5.

\section{Preliminaries}

Throughout the article, we assume that $\mathcal{H}_X$ is a separable Hilbert space describing a quantum system $X$,
$\mathfrak{B}_{\rm sa}(\mathcal{H}_X)$ is the real Banach space of all Hermitian bounded operators on $\mathcal{H}_X$ with the operator norm $\|\cdot\|$ and $\mathfrak{T}(\mathcal{H}_X)$ is the (complex)
Banach space of all trace-class
operators on $\mathcal{H}_X$  with the trace norm $\|\!\cdot\!\|_1$. Let
$\mathfrak{S}(\mathcal{H}_X)$ be  the set of quantum states (positive operators
in $\mathfrak{T}(\mathcal{H}_X)$ with unit trace) \cite{H-SCI,Wilde}.\smallskip

Write $I_{X}$ for the unit operator on a Hilbert space
$\mathcal{H}_X$.\smallskip

The \emph{von Neumann entropy} of a quantum state
$\rho \in \mathfrak{S}(\HH)$ is  defined by the formula
$S(\rho)=\operatorname{Tr}\eta(\rho)$, where  $\eta(x)=-x\ln x$ if $x>0$
and $\eta(0)=0$. It is a concave lower semicontinuous function on the set~$\mathfrak{S}(\HH)$ taking values in~$[0,+\infty]$ \cite{H-SCI,L-2,W}.

We will use the  homogeneous extension of the von Neumann entropy to the positive cone $\T_+(\HH)$ defined as
\begin{equation}\label{S-ext}
S(\rho)\doteq(\Tr\rho)S(\rho/\Tr\rho)=\Tr\eta(\rho)-\eta(\Tr\rho)
\end{equation}
or any nonzero operator $\rho$ in $\T_+(\HH)$ and equal to $0$ at the zero operator \cite{L-2}.\smallskip

The concavity of the von Neumann entropy on $\St(\HH)$ implies that
\begin{equation}\label{S-in}
S(\rho)+S(\sigma)\leq S(\rho+\sigma)\quad \forall \rho,\sigma\in\T_+(\HH).
\end{equation}

A finite or
countable collection $\{\rho_{i}\}$ of states
with a probability distribution $\{p_{i}\}$ is conventionally called
\textit{discrete ensemble} and denoted by $\{p_{i},\rho_{i}\}$. The state
$\bar{\rho}\doteq\sum_{i}p_{i}\rho_{i}$ is called \emph{average state} of this  ensemble. \smallskip

A \textit{generalized (continuous) ensemble} is defined as
a Borel probability measure on the set of quantum states \cite{H-SCI,H-Sh-2}. We denote by $\PP(\mathcal{H})$ the set of all Borel probability measures on $\ext\mathfrak{S}(\mathcal{H})$ equipped with the topology of weak convergence
\cite{Bog,Par}.\footnote{The weak convergence of a sequence $\,\{\mu_n\}\subset\PP(\mathcal{H})\,$  to a measure $\,\mu_0\in\PP(\mathcal{H})\,$ means that
$\,\lim_{n\rightarrow\infty}\int f(\varrho)\mu_n(d\varrho)=\int f(\varrho)\mu_0(d\varrho)\,$
for any continuous bounded function $f$ on $\,\St(\HH)$.}
 The set $\mathcal{P}(\mathcal{H})$ is a complete
separable metric space containing the dense subset $\PP_0(\mathcal{H})$ of discrete measures (corresponding to discrete ensembles) \cite{Par}. The average state of a generalized
ensemble $\mu \in \PP(\mathcal{H})$ is the barycenter of the measure
$\mu $ defined by the Bochner integral
\begin{equation*}
\bar{\rho}(\mu )=\int_{\ext\mathfrak{S}(\mathcal{H})}\rho \mu (d\rho ).
\end{equation*}

We will use the function $S_A$ on the set $\PP(\mathcal{H}_{AB})$ defined as
\begin{equation}\label{SA-def}
  S_A(\mu)\doteq\int_{\ext\St(\HH_{AB})}S(\rho_A)\mu(d\rho).
\end{equation}
The  lower semicontinuity of the von Neumann entropy implies the  lower semicontinuity of the function $S_A$ on the set $\PP(\mathcal{H}_{AB})$ \cite{H-Sh-2}.
Using this function the definition (\ref{EF-c}) can be rewritten as
\begin{equation}\label{EF}
  E_F(\rho)=\inf_{\mu\in\PP_\rho(\mathcal{H}_{AB})}S_A(\mu),
\end{equation}
where $\PP_\rho(\mathcal{H}_{AB})$ is the subset of $\PP(\HH_{AB})$ consisting of measures $\mu$ such that $\,\bar{\rho}(\mu)=\rho$.

\section{Criterion for attainability of the supremum in (\ref{VE})}

In this section, we consider (generalizing  the task mentioned before) a criterion for attainability of the supremum in the expression
\begin{equation}\label{VE-0}
  E_F(\rho)=\sup\left\{\Tr\Lambda\rho\,\left|\,\Lambda\in\Bsa(\HH_{0}),
  \langle\psi|\Lambda|\psi\rangle\le S\bigl(\psi_A\bigr)\ \forall\psi\in \HH^1_{0}\right.\right\},
\end{equation}
where $\rho$ is an arbitrary state in $\St(\HH_{AB})$, $\HH_0$ is any subspace of $\HH_{AB}$ containing the support of $\rho$,
$\Bsa(\HH_{0})$ is the set of all bounded Hermitian operators on $\HH_{0}$ and $\HH^1_{0}$ is the unit sphere in $\HH_{0}$. This expression
is proved in the same way as the expression in (\ref{VE}) by noting that the function $E_F$ on $\St(\HH_{0})$ coincides with
the convex closure of the the concave function $\varrho\mapsto S(\varrho_A)$ on $\St(\HH_{0})$ and by using the Fenchel-Moreau theorem.

\begin{theorem}\label{P-3} Let $A$ and $B$ be quantum systems of any dimensions. Let $\rho$ be a state of the system $AB$ and $\HH_0$ be a subspace of $\HH_{AB}$ such that $\,\supp\rho\subseteq\HH_0$.
Write $D(\Lambda)$ for the diameter of the spectrum of an operator $\Lambda\in\B_{\rm sa}(\HH_0)$:\smallskip
\begin{equation}\label{D-def}
D(\Lambda)=\max \mathrm{Sp}(\Lambda)-\min \mathrm{Sp}(\Lambda).
\end{equation}
For a given $C>0$ the following properties are equivalent:
\begin{enumerate}
  \item[\emph{(i)}] the inequality $\,E_F(\rho)-E_F(\sigma)
        \leq C\|\rho-\sigma\|_1\,$ holds for any state $\,\sigma\in\St(\mathcal{H}_{0})$;
  \item[\emph{(ii)}] there exists an operator
    $\Lambda_\rho\in\B_{\mathrm{sa}}(\mathcal{H}_0)$ such that  $D(\Lambda_\rho)\leq 2C$,
    \[
        E_F(\rho)
        =\operatorname{Tr}\Lambda_\rho\rho
        \quad and\quad
        \langle\psi|\Lambda_\rho|\psi\rangle
        \leq S(\psi_A)
        \quad
        \forall\,\psi\in\mathcal{H}^1_{0};
    \]
    \item[\emph{(iii)}] there exists an operator
    $\widetilde{\Lambda}_\rho\in\B_{\mathrm{sa}}(\mathcal{H}_{0})$ such that  $\,\|\widetilde{\Lambda}_\rho\|\leq C\,$ and
    \[
        E_F(\rho)=\inf_{\psi\in\mathcal{H}^1_{0}}
        \left\{
            S(\psi_A)+\operatorname{Tr}\widetilde{\Lambda}_\rho(\rho-\psi)\right\}.
    \]
 \end{enumerate}
\end{theorem}

\begin{remark}\label{P-3r} Property $\rm (ii)$ in Theorem \ref{P-3} can be strengthened by saying that
$$
\langle\psi|\Lambda_\rho|\psi\rangle=S(\psi_A)\quad\textrm{ for }\mu_\rho\textrm{-almost all}\;\;\psi=|\psi\rangle\langle\psi|\in\ext\St(\mathcal{H}_0),
$$
where $\mu_\rho$ is any measure on $\ext\St(\mathcal{H}_{0})$  such that $\int\varrho\mu(d\varrho)=\rho$ and  $E_F(\rho)=\int\! S(\varrho_A)\mu(d\varrho)$.\footnote{Existence of this measure for any state $\rho$ is proved in \cite{EM}.} Indeed, for any such measure we have
$$
\int_{\St(\HH_0)}\left[S(\psi_A)-\langle\psi|\Lambda_\rho|\psi\rangle\right]\mu_\rho(d\psi)=0.
$$
Since the expression in the squire brackets is nonnegative, it is equal to zero for\break $\mu_\rho$-almost all $\psi$.
\end{remark}
\smallskip
\emph{Proof.} $\rm (i)\Rightarrow(iii)$.  Note that
\begin{equation}\label{meq}
\begin{aligned}
E_F(\rho) & =\!\inf_{\sigma\in\St(\HH_{0})}\!\left\{E_F(\sigma)+C\|\rho-\sigma\|_1\right\}
 \\   &
=\!\inf_{\mu\in\PP(\HH_{0})}\!\left\{S_A(\mu)+C\|\rho-\bar{\rho}(\mu)\|_1\right\}\\ &
=\!\inf_{\mu\in\PP(\HH_{0})}\max_{\Lambda\in\B(C)}G(\mu,\Lambda)=\max_{\Lambda\in\B(C)}\inf_{\mu\in\PP(\HH_{0})}G(\mu,\Lambda),
\end{aligned}
\end{equation}
where $S_A$ is the function defined in (\ref{SA-def}), $G(\mu,\Lambda)\doteq S_A(\mu)+\Tr \Lambda(\rho-\bar{\rho}(\mu))$ and $\B(C)$ is the ball of radius $C$ in $\Bsa(\HH_0)$. The first  equality in (\ref{meq})  holds due to the validity of $\rm(i)$ (as $\rho\in\St(\HH_0)$).  The second equality in (\ref{meq})  follows from the expression (\ref{EF}) and the possibility the permute the infima\footnote{i.e. the property $\inf_{x\in X}\inf_{y\in Y}f(x,y)=\inf_{y\in Y}\inf_{x\in X}f(x,y)$.}  because it is clear that
$\PP_\rho(\mathcal{H}_{AB})\subset\PP(\HH_{0})$.  The third equality in (\ref{meq})  is obvious. The last equality in (\ref{meq})  follows
from Theorem 3.1 in \cite{sim}, since the function
$G(\mu,\Lambda)$ is affine in both arguments and continuous in $\Lambda$ w.r.t. the $\sigma$-weak operator topology on $\B_{\rm sa}(\HH_{0})$, while the ball $\B(C)$ is
compact in this topology.

Thus, to prove $\rm(iii)$ it suffices to show that
$$
\inf_{\mu\in\PP(\HH_{0})}G(\mu,\Lambda)=\inf_{\psi\in\HH^1_{0}}H(\psi,\Lambda),\quad \textrm{where} \quad H(\psi,\Lambda)\doteq S(\psi_A)+\Tr \Lambda(\rho-\psi).
$$
This can be done by combining the following facts:
\begin{itemize}
  \item the set $\ext\PP(\HH_{0})$ of extreme points of the convex set $\PP(\HH_{0})$ is closed and consists of Dirac (single atom) measures concentrated at
  the pure states in $\St(\HH_{0})$ which correspond to vectors in $\HH^1_{0}$;
  \item any measure $\mu$ in $\PP(\HH_{0})$ can be represented as
  $$
  \mu=\rm \beta(\textbf{p})\doteq\int_{\ext\PP(\HH_{0})} \nu \textbf{p}(d\nu)
  $$
  for some probability measure $\textbf{p}$ on $\ext\PP(\HH_{0})$ (this can be proved, for instance, by using the $\mu$-compactness of the set $\PP(\HH_{0})$ (in terms of \cite{P&Sh}) and
  Proposition 5.1 in \cite{P&Sh});
  \item the affine and bounded from below function $\,\mu\mapsto G(\mu,\Lambda)\,$ is lower semicontinuous on $\PP(\HH_{0})$ and hence the equality
  $$
  G(\rm \beta(\textbf{p}),\Lambda)=\int_{\ext\PP(\HH_{0})}G(\nu,\Lambda)\textbf{p}(d\nu)
  $$
  holds for any probability measure $\textbf{p}$ on $\ext\PP(\HH_{0})$ (this can be shown, for instance, by using the arguments from the proof
  of Corollary A-1 in \cite[the Appendix]{EM});
\end{itemize}

$\rm (iii)\Rightarrow(ii)$. This implication is proved by noting that $\,H(\psi,\Lambda)=H(\psi,\Lambda+cI_{0})\,$ for any $c\in\mathbb{R}$, where $I_0$ is the unit operator on $\HH_0$. Indeed, the last property implies
\begin{equation}\label{tpt}
\begin{aligned}
E_F(\rho)=\max_{\Lambda\in \B(C)}
\inf_{\psi\in\mathcal{H}^1_{0}}
H(\psi,\Lambda)
&=
\max\left\{
    \operatorname{Tr}\Lambda\rho
    \,\middle|\,
    \Lambda\in \mathfrak{D}(C),\;
    g(\Lambda)=0
\right\},
\end{aligned}
\end{equation}
where the first equality follows from $\rm (iii)$,
\[
    \mathfrak{D}(C)
    =
    \left\{
        \Lambda\in\mathcal{B}_{\mathrm{sa}}(\mathcal{H}_{0})
        \,\left|\,
         D(\Lambda)\leq 2C
    \right.\right\}\quad \textrm{and} \quad  g(\Lambda)=\inf_{\psi\in\mathcal{H}^1_{0}}\left\{S(\psi_A)-\langle\psi|\Lambda|\psi\rangle
    \right\}.
\]

To complete the proof of $\rm(ii)$ it suffices to show  that the r.h.s. of (\ref{tpt}) is equal to\break $X\doteq\sup\left\{
        \operatorname{Tr}\Lambda\rho
        \,\middle|\,
        \Lambda\in \mathfrak{D}(C),\;
        g(\Lambda)\geq 0
    \right\}$.
Denote the r.h.s. of (\ref{tpt}) by $Y$. It is clear that $Y\leq X$. Assume that $Y<X$.
Then there is an operator $\Lambda\in \mathfrak{D}(C)$ such that $g(\Lambda)>0$ and
$\operatorname{Tr}\Lambda\rho-Y=\Delta>0$. But in this case
\[
  \Lambda+\Delta I_0\in\mathfrak{D}(C),\quad   g(\Lambda+\Delta I_0)=0\quad \textrm{and}\quad
    \operatorname{Tr}(\Lambda+\Delta I_0)\rho
    =
    \operatorname{Tr}\Lambda\rho+\Delta,
\]
contradicting our assumption.\smallskip

$\rm (ii)\Rightarrow(i)$.  It follows from $\rm (ii)$ that
$$
S_A(\nu)=\int_{\ext\St(\HH_0)}S(\psi_A)\nu(d\psi)\geq \int_{\ext\St(\HH_0)}\langle\psi|\Lambda_\rho|\psi\rangle\nu(d\psi)=\operatorname{Tr}\Lambda_\rho\bar{\rho}(\nu)
$$
for any $\nu\in\PP(\HH_0)$.  Since the set of measures in $\PP(\HH_{AB})$ with barycenters in $\St(\HH_0)$ coincides with the set $\PP(\HH_0)$, this inequality and the definition of the function $E_F$ imply
\[
E_F(\sigma)
\geq
\operatorname{Tr}\Lambda_\rho\sigma,
\qquad
\forall\,\sigma\in \St(\mathcal{H}_0),
\]
and therefore
\[
E_F(\rho)-E_F(\sigma)
\leq
\operatorname{Tr}\Lambda_\rho(\rho-\sigma)
\leq
C\|\rho-\sigma\|_1,
\]
where the last inequality follows from the well known continuity bound for the function $\,\vartheta\mapsto\Tr\Lambda_\rho\vartheta\,$ (see, for instance, Example 2 in \cite{QC}), since $\, D(\Lambda_\rho)\leq 2C$. $\Box$

\bigskip

\smallskip
Theorem 1 with $\mathcal{H}_0=\mathcal{H}_{AB}$ gives a criterion
for existence of a (global) supporting functional for the EoF at the state
$\rho$.

Theorem 1 with $\mathcal{H}_0=\operatorname{supp}\rho$ gives a criterion
for existence of so-called \emph{local supporting functional} for the EoF at the
state $\rho$, i.e. such a functional
$\,\ell_\rho(\sigma)=\operatorname{Tr} \Lambda_\rho\sigma,$ $\Lambda_\rho\in\Bsa(\operatorname{supp}\rho)$, on $\,\St(\operatorname{supp}\rho)\,$
that
\[
\ell_\rho(\rho)=E_F(\rho)\quad \textrm{and}\quad \ell_\rho(\sigma)\leq E_F(\sigma)\quad\forall\sigma\in \St(\operatorname{supp}\rho).
\]

The general results of convex analysis described in the Introduction guarantee the existence
of a local supporting functional for the EoF at an arbitrary state $\rho$ of a
finite-dimensional system $AB$. Proposition \ref{cex} in the next section
shows that\emph{ for  some states $\rho$  of 2-qubit system local supporting functional for the EoF cannot be
extended to global supporting functional for the EoF.}

\section{Basic counterexample and its corollaries}

It is mentioned in the Introduction  that  general results of convex analysis guarantee the existence
of a (global) supporting functional for the EoF at an arbitrary non-degenerate state $\rho$ of any finite-dimensional bipartite quantum system.

The aim of this section is to show that a supporting functional for the EoF \emph{may not exist} in some degenerate states
of \emph{all}  bipartite quantum systems of any dimensions. The crucial result here is the following corollary of the implication $\,\rm(ii)\Rightarrow(i)\,$ in Theorem \ref{P-3} obtained with the help of Claude Fable 5.

\begin{proposition}\label{cex} Let $A$ and $B$ be qubit quantum systems. Let
$|00\rangle,|01\rangle,|10\rangle,|11\rangle$ be a  basis in $\HH_{AB}=\mathbb C^2\otimes\mathbb C^2$. The EoF has no supporting functional at the state
\begin{equation}\label{eq:omega}
  \rho=\frac{1}{2}|\Phi^+\rangle\langle\Phi^+|+\frac{1}{2}|01\rangle\langle01|,
\end{equation}
where  $|\Phi^+\rangle=\frac{1}{\sqrt2}(|00\rangle+|11\rangle)$.
\end{proposition}

\textbf{Note:} The state $\rho$  has the spectrum
$\{\frac{1}{2},\frac{1}{2},0,0\}$ and, hence, $\rank\rho=2$.\smallskip

\emph{Proof.}
By the implication $\,\rm (ii)\Rightarrow(i)\,$ in Theorem 1 with $\,\HH_0=\HH_{AB}\,$ it suffices to find
a state $\sigma$ such that
\begin{equation}\label{eq:lin}
  \lim_{t\to0^+}\frac{E_F(\rho)-E_F(\rho_t)}{t}=+\infty,
\end{equation}
where $\rho_t=(1-t)\rho+t\sigma$, $\,t\in[0,1]$.\smallskip

We will show that (\ref{eq:lin}) holds with $\sigma=|10\rangle\langle10|$.

For two qubits, $E_F$ is given  by Wootters' formula~\cite{Wootters,HillWootters}:
\begin{equation}\label{eq:wootters}
  E_F(\rho)=\mathcal E\bigl(C(\rho)\bigr),
  \qquad
  \mathcal E(C)=h\Bigl(\tfrac12\bigl(1+\sqrt{1-C^2}\bigr)\Bigr),
\end{equation}
where $h$ is the binary entropy, $C(\rho)=\max\{0,\lambda_1-\lambda_2-\lambda_3-\lambda_4\}$ and
$\lambda_1\ge\dots\ge\lambda_4$ are the square roots of the eigenvalues of $\rho\widetilde\rho$, with the spin flip
$\widetilde\rho=(\sigma_y\otimes\sigma_y)\,\overline\rho\,(\sigma_y\otimes\sigma_y)$.

\begin{lemma}\label{lem:concurrence}
For the state $\rho_t$ defined after \eqref{eq:lin}, for $0\le t\le t_0=1/3$,
\begin{equation}\label{eq:Cexact}
  C(\rho_t)=\frac{1-t}{2}-2\sqrt{\frac{t(1-t)}{2}}
  =\frac{1}{2}-\sqrt{2t}+O(t),\qquad t\ra0^+.
\end{equation}
In particular, $C(\rho)=\tfrac12$ and
$E_F(\rho)=\mathcal E(\tfrac12)\approx0.246$.
\end{lemma}

\begin{proof}
The operator $\rho_t$ is block diagonal with respect to
$\HH=\mathcal K_1\oplus\mathcal K_2$, where
$\mathcal K_1=\operatorname{span}\{|00\rangle,|11\rangle\}$ and
$\mathcal K_2=\operatorname{span}\{|01\rangle,|10\rangle\}$. The spin flip preserves each subspace (in the ordered basis $|00\rangle,|01\rangle,|10\rangle,|11\rangle$ its unitary part is the anti-diagonal matrix with entries $(-1,1,1,-1)$), so $\rho_t\widetilde\rho_t$ is also block diagonal.\smallskip

\textbf{Block $\mathcal K_1$.}
Here
$\rho_t|_{\mathcal K_1}=a\,uu^*$ with
$u=\tfrac1{\sqrt2}(1,1)^{\mathsf T}$ and $a=\tfrac{1-t}{2}$. The spin flip restricted to this block acts by complex conjugation followed by exchange of the two coordinates, with an overall sign. Since $\sigma_xu=u$ and $u$ is real,
$\widetilde\rho_t|_{\mathcal K_1}=\rho_t|_{\mathcal K_1}$. Hence
$(\rho_t\widetilde\rho_t)|_{\mathcal K_1}=a^2uu^*$, with eigenvalues $a^2$ and $0$. This contributes $\lambda=\tfrac{1-t}{2}$ and $\lambda=0$.\smallskip

\textbf{Block $\mathcal K_2$.}
Here
$\rho_t|_{\mathcal K_2}=\operatorname{diag}(\tfrac{1-t}{2},t)$ in the basis $|01\rangle,|10\rangle$. The spin flip interchanges $|01\rangle\leftrightarrow|10\rangle$, so
$\widetilde\rho_t|_{\mathcal K_2}=\operatorname{diag}(t,\tfrac{1-t}{2})$ and
\[
  (\rho_t\widetilde\rho_t)|_{\mathcal K_2}
  =\frac{t(1-t)}{2}\,P_2,
\]
where $P_2$ is the projector on $\mathcal K_2$.
This gives $\lambda_{2,3}=\sqrt{t(1-t)/2}$.

Thus, for $t\le t_0$,
$\lambda_1=\tfrac{1-t}{2}$,
$\lambda_2=\lambda_3=\sqrt{t(1-t)/2}$,
$\lambda_4=0$, and Wootters' formula yields \eqref{eq:Cexact}. The expansion follows from
$\sqrt{2t(1-t)}=\sqrt{2t}\,(1-t/2+O(t^2))$.
\end{proof}

\begin{lemma}\label{lem:slope}
The function $\mathcal E$ in \eqref{eq:wootters} is smooth and strictly increasing on $(0,1)$, and
\begin{equation}\label{eq:kappa}
  \kappa:=\mathcal E'\bigl(\tfrac12\bigr)
  =\frac{\ln(2+\sqrt3)}{\sqrt3}
  \approx0.760>0.
\end{equation}
\end{lemma}

\begin{proof}
Let $x=x(C)=\tfrac12(1+\sqrt{1-C^2})$, so that $\mathcal E=h\circ x$. Then
$h'(x)=\ln\frac{1-x}{x}$ and
$x'(C)=-\dfrac{C}{2\sqrt{1-C^2}}$. At $C=\tfrac12$,
$\sqrt{1-C^2}=\tfrac{\sqrt3}{2}$,
$x=\tfrac{2+\sqrt3}{4}$, and
$1-x=\tfrac{2-\sqrt3}{4}$. Since
$(2+\sqrt3)(2-\sqrt3)=1$,
we obtain
$\frac{1-x}{x}=\frac{2-\sqrt3}{2+\sqrt3}=(2-\sqrt3)^2$,
so
$h'(x)=2\ln(2-\sqrt3)=-2\ln(2+\sqrt3)$, while
$x'(\tfrac12)=-\tfrac1{2\sqrt3}$. Their product gives \eqref{eq:kappa}.
\end{proof}

By Lemmas~\ref{lem:concurrence} and~\ref{lem:slope}, and by smoothness of $\mathcal E$ near $C=\frac{1}{2}$,
\begin{equation}\label{eq:drop}
  E_F(\rho)-E_F(\rho_t)
  =\mathcal E\left(\frac{1}{2}\right)-\mathcal E\left(\frac{1}{2}-\sqrt{2t}+O(t)\right)
  =\kappa\sqrt2\,\sqrt t+O(t),
\end{equation}
where $\kappa\sqrt2\approx1.07$. This implies (\ref{eq:lin}) and proves the nonexistence of supporting affine functional for the EoF at the state $\rho$. $\Box$\medskip

\begin{corollary}\label{cex-c} In a bipartite quantum system $AB$ of any dimension there exist degenerate states at which the EoF has no supporting affine functional.
\end{corollary}

\emph{Proof.} Let $\HH^0_A$ and $\HH^0_B$ be arbitrary two-dimensional subspaces of  $\HH_A$ and $\HH_B$, respectively. By
Proposition \ref{cex} there is a state $\rho$ in $\St(\HH^0_A\otimes\HH^0_B)$ such that the EoF  has no  supporting functional at $\rho$
as a function on $\St(\HH^0_A\otimes\HH^0_B)$. If we treat $\rho$ as a state in $\St(\HH_A\otimes\HH_B)$ and assume that the EoF  has a supporting functional at $\rho$
as a function on $\St(\HH_A\otimes\HH_B)$ then the restriction of this functional to the convex set $\St(\HH^0_A\otimes\HH^0_B)$  will be a  supporting functional for the EoF (treated as a function on $\St(\HH^0_A\otimes\HH^0_B)$) at the state $\rho$ contradicting the choice of $\rho$. $\Box$

\section{On Lipschitz lower semicontinuity of the EoF}

The implication $\,\rm(ii)\Rightarrow(i)\,$ in Theorem \ref{P-3} gives a way to prove, for given state\break $\rho\in\St(\HH_{AB})$ and subspace $\,\HH_0\subseteq\HH_{AB},$
the inequality of the form
$$
E_F(\rho)-E_F(\sigma)
        \leq C_{\rho,\HH_0}\|\rho-\sigma\|_1\,\quad \forall\sigma\in\St(\mathcal{H}_{0}),
$$
where $C_{\rho,\HH_0}$ is a positive number (as small as possible) depending on $\rho$ and $\HH_0$.
This inequality in case $\HH_0=\HH_{AB}$ (resp. in the case $\HH_0\subset\HH_{AB}$) can be called  unrestricted (resp. $\HH_0$-restricted)  Lipschitz lower semicontinuity bound for the
EoF at a state $\rho$.

Theorem \ref{P-3} reduces the above task to the problem of
achievability of the supremum in the expressions:
\begin{equation}\label{2ex}
    E_F(\rho)
    =
    \sup_{\Lambda\in\mathcal{B}_{\mathrm{sa}}(\mathcal{H}_0)}
    \inf_{\psi\in\mathcal{H}^1_{0}}
    \left\{
        S(\psi_A)
        +
        \operatorname{Tr}\Lambda
        \bigl(\rho-\psi\bigr)
    \right\}=\sup_{\Lambda\in\mathfrak{B}_S}\Tr \Lambda\rho,
\end{equation}
where
$\,\mathfrak{B}_S\doteq\left\{\Lambda \in\B_{\rm sa}(\HH_{0})\,\left|\, \langle\psi|\Lambda|\psi\rangle
        \leq S(\psi_A)\;\, \forall\psi\in\HH^1_{0}\right.\right\}$.\smallskip

Below, we present the results of two types obtained by using Theorem \ref{P-3}.

\begin{proposition}\label{LSC-1} Let $A$ and $B$ be quantum systems of any dimensions. Let $\rho$ be a finite-rank state of the system $AB$ and $\lambda_{\min}^{\rho}$ be the minimal eigenvalue of $\rho$.\smallskip

\emph{A)} If $\,\rank\rho_A=d<+\infty$, then
\[
E_F(\rho)-E_F(\sigma)
\leq
\left[\frac{\log d-E_F(\rho)}
{\lambda_{\min}^{\rho}}\right]
\frac{\|\rho-\sigma\|_1}{2}
\]
for any $\,\sigma\in \St(\HH_{AB})\,$ such that $\,\supp \sigma\subseteq\supp\rho$.\smallskip

\emph{B)} If $\,S(\rho_A)<+\infty,$ then
\[
E_F(\rho)-E_F(\sigma)
\leq
\left[
\frac{S(\rho_A)}
{(\lambda_{\min}^{\rho})^2}
-
\frac{E_F(\rho)}
{\lambda_{\min}^{\rho}}
\right]
\frac{\|\rho-\sigma\|_1}{2}
\]
for any $\,\sigma\in \St(\HH_{AB})\,$ such that $\,\supp \sigma\subseteq\supp\rho$.
\end{proposition}

\noindent
\textit{Proof.} Both claims follow from the implication
$\rm (ii)\Rightarrow(i)$ in Theorem \ref{P-3} with $\,\HH_0=\supp\rho\,$ and Lemma \ref{lll} below.
Indeed, if $\,\rank\rho_A=d<+\infty$, then $\,S(\psi_A)\leq\log d\,$ for all  $\,\psi\in \supp\rho$.
If $S(\rho_A)<+\infty$, then $\,S(\psi)\leq (1/\lambda_{\min}^{\rho})S(\rho_A)\,$ for all  $\,\psi\in \supp\rho\,$ due to inequality (\ref{S-in})
(as $\,\psi\leq (1/\lambda_{\min}^{\rho})\rho$). $\Box$

\begin{lemma}\label{lll}
Let $\rho$ be a finite-rank state of the system $AB$ such that
\[
S(\psi_A)\leq C<+\infty
\qquad \forall\,\psi\in \HH_\rho\doteq\supp\rho.
\]
Then property $\rm(ii)$ in Theorem \ref{P-3} holds with
$\,\HH_0=\HH_\rho\,$ and an operator
$\,\Lambda_\rho\in \B_{\rm sa}(\HH_\rho)$ such that
\[
D(\Lambda_\rho)
\leq
\frac{C-E_F(\rho)}{\lambda_{\min}^\rho}.
\]
\end{lemma}
\bigskip

\textit{Proof.} Take $\delta>0$ and consider the closed set $\B_\delta$ of Hermitian
operators $\Lambda$ on $\HH_\rho$ such that
\begin{equation}\label{str}
\Tr\Lambda\rho
\in
[E_F(\rho)-\delta,E_F(\rho)]
\quad \textrm{and}\quad
\langle\psi|\Lambda|\psi\rangle\leq S(\psi_A),
\qquad \forall\,\psi\in \HH^1_\rho.
\end{equation}

For a given $\Lambda\in \B_\delta$ consider the state
$\sigma
=
\rho+
\lambda_{\min}^{\rho}
\bigl(
|\alpha\rangle\langle\alpha|
-
|\beta\rangle\langle\beta|
\bigr),
$
where $\alpha$ (resp. $\beta$) is the unit eigenvector of $\Lambda$
corresponding to the maximal (resp. minimal) eigenvalue of $\Lambda$.
Then the properties in  (\ref{str}) imply $\,E_F(\sigma)\geq \Tr\Lambda \sigma\,$
and $\,E_F(\rho)\leq\Tr\Lambda\rho+\delta$.
Hence
$$
E_F(\sigma)-E_F(\rho)
\geq
\lambda_{\min}^{\rho}
\Tr\Lambda
\bigl(
|\alpha\rangle\langle\alpha|
-
|\beta\rangle\langle\beta|
\bigr)-\delta
=
\lambda_{\min}^{\rho}D(\Lambda)-\delta.
$$

So,
\[
D(\Lambda)
\leq
\frac{E_F(\sigma)-E_F(\rho)+\delta}
{\lambda_{\min}^{\rho}}
\leq
\frac{C-E_F(\rho)+\delta}
{\lambda_{\min}^{\rho}}.
\]

Since the second property in (\ref{str}) shows that
\[
\langle\alpha|\Lambda|\alpha\rangle
=\max \mathrm{Sp}(\Lambda)\leq C\qquad \forall \Lambda\in\B_\delta,
\]
the set $\B_\delta$ lies within some ball in
$\B_{\rm sa}(\HH_\rho)$ and hence it is compact.

This implies that the second supremum in (\ref{2ex}) is attained at some operator
$\Lambda_\rho$ in $\B_S\cap\B_0$, i.e. property (ii) in Theorem \ref{P-3} with
$\HH_0=\HH_\rho$ holds for the state $\rho$ and this operator $\Lambda_\rho$.

\begin{proposition}\label{LSC-2} Let $A$ and $B$ be quantum systems of any dimensions. Let $\rho$ be a finite-rank state of the system $AB$ such that
$\,\HH_\rho\doteq\supp\rho=\supp\rho_A\otimes\supp\rho_B\,$ and $\,\lambda_{\min}^{\rho}\,$ be the minimal eigenvalue of $\rho$.\smallskip

\emph{A)} If $\,\rank\rho_A=d<+\infty$, then
\[
E_F(\rho)-E_F(\sigma)
\leq
\left[\frac{\log d-E_F(\rho)}
{\lambda_{\min}^{\rho}}\right]
\frac{\|\rho-\sigma\|_1}{2}
\]
for any $\,\sigma\in \St(\HH_{AB})$.\smallskip

\emph{B)} If $\,S(\rho_A)<+\infty,$ then
\[
E_F(\rho)-E_F(\sigma)
\leq
\left[
\frac{S(\rho_A)}
{(\lambda_{\min}^{\rho})^2}
-
\frac{E_F(\rho)}
{\lambda_{\min}^{\rho}}
\right]
\frac{\|\rho-\sigma\|_1}{2},
\]
for any $\,\sigma\in \St(\HH_{AB})$.
\end{proposition}\smallskip

\emph{Proof.} By the proof of Proposition \ref{LSC-1}  we have
\[
S(\psi_A)\leq C<+\infty
\qquad \forall\,\psi\in \HH_\rho\doteq\supp\rho,
\]
where $C=\log d\,$ in case $A$ and $C=(1/\lambda_{\min}^{\rho})S(\rho_A)\,$ in case $B$.
\medskip

In both cases Lemma \ref{lll} shows the existence of an operator
$\Lambda_\rho\in \B_{\rm sa}(\HH_\rho)$ such that
\begin{equation}\label{fff}
E_F(\rho)=\Tr\Lambda_\rho\rho,
\quad \quad
\langle\psi|\Lambda_\rho|\psi\rangle\leq S(\psi_A)
\qquad \forall\,\psi\in \HH^1_\rho
\end{equation}
and
\begin{equation}\label{fgf}
D(\Lambda_\rho)\leq\frac{C-E_F(\rho)}
{\lambda_{\min}^{\rho}}.
\end{equation}

Note that the minimal eigenvalue of $\Lambda_\rho$ is negative. Indeed, if we
assume that $\Lambda_\rho\geq 0$, then we obtain
\[
E_F(\bar{\rho})\geq \Tr \Lambda_\rho\bar{\rho}>0,
\]
where $\,\bar{\rho}\doteq (1/r)\rho^0,$ $r=\rank \rho$, is the chaotic state in $\HH_\rho$. This is not possible, since $\,E_F(\bar{\rho})=0\,$ (as $\,\bar{\rho}=\bar{\rho}_A\otimes\bar{\rho}_B$), where $\bar{\rho}_X$
is the chaotic state in $\St(\supp\rho_X)$, $X=A,B$.

Consider the operator $\,\Lambda_\rho^{\ext}=P_\rho\Lambda_\rho P_\rho$ on $\HH_{AB}$,
where $P_\rho$ is the projector onto $\HH_\rho$. Note that $P_\rho=Q_\rho^A\otimes Q_\rho^B$, where $Q_\rho^X$ is the projector onto $\supp\rho_X$, $X=A,B$.
It is clear that
\[
\Tr\Lambda_\rho^{\ext}\rho=E_F(\rho).
\]
To show that
\[
\langle\psi|\Lambda_\rho^{\ext}|\psi\rangle\leq S(\psi_A)
\qquad \forall\,\psi\in \HH^1_{AB},
\]
note that
\begin{align*}
[P_\rho\psi]_A
&=
\Tr_B
\bigl(
Q_\rho^A\otimes Q_\rho^B
|\psi\rangle\langle\psi|
Q_\rho^A\otimes Q_\rho^B
\bigr)
\\
&=
Q_\rho^A
\Bigl[
\Tr_B
\bigl(
I_A\otimes Q_\rho^B
|\psi\rangle\langle\psi|
\bigr)
\Bigr]
Q_\rho^A
\\
&\leq
Q_\rho^A
\bigl[
\Tr_B|\psi\rangle\langle\psi|
\bigr]
Q_\rho^A,
\end{align*}
where $I_A$ is the unit operator on $\HH_A$. Hence inequality (\ref{S-in}) and Lemma 3 in \cite{L-2} show that
\[
S\bigl([P_\rho\psi]_A\bigr)
\leq S\bigl(Q_\rho^A
\bigl[
\Tr_B|\psi\rangle\langle\psi|
\bigr]
Q_\rho^A\bigr)
\leq S(\psi_A)
\qquad \forall\,\psi\in \HH^1_{AB},
\]
where $S$ in the l.h.s. is the extension
of the von Neumann entropy defined in (\ref{S-ext}). Thus, we have
\[
\langle\psi|\Lambda_\rho^{\ext}|\psi\rangle
=\langle P_\rho\psi|\Lambda_\rho|P_\rho\psi\rangle
\leq S\bigl([P_\rho\psi]_A\bigr)
\leq S(\psi_A),\qquad \forall\,\psi\in \HH^1_{AB},
\]
where the first inequality follows from the second property in (\ref{fff}).

Thus, property $\rm (ii)$ in Theorem \ref{P-3}  with
$\HH_0=\HH_{AB}$ holds for the state $\rho$ and the operator $\Lambda_\rho^{\ext}$ in $\B_{\rm sa}(\HH_{AB})$.
Since the minimal eigenvalue of $\Lambda_\rho$ is negative, $D(\Lambda^{\ext}_\rho)=D(\Lambda_\rho)$.
Thus, both claims of the proposition follow from the implication $\rm (ii)\Rightarrow(i)$ in Theorem \ref{P-3} and (\ref{fgf}). $\Box$
\medskip

The proof of Proposition \ref{LSC-2} shows how starting from a  local supporting functional for the EoF
at a state from a certain class to construct a global supporting functional for the EoF at this state.


\end{document}